\documentclass[letterpaper]{article} 
\usepackage[preprint]{aaai23}  
\usepackage{times}  
\usepackage{helvet}  
\usepackage{courier}  
\usepackage[hyphens]{url}  
\usepackage{graphicx} 
\usepackage{natbib}  
\usepackage{caption} 

\usepackage{booktabs}
\usepackage{tabularx}
\usepackage[dvipsnames]{xcolor}
\usepackage{tcolorbox}
\usepackage{multirow}

\usepackage{centernot}
\usepackage{algorithm2e}
\usepackage{comment}
\usepackage{enumerate}
\usepackage{amssymb}
\usepackage{amsmath}
\usepackage{amsthm}
\usepackage{subcaption}
\usepackage{multicol}
\usepackage{makecell}

\newtheorem{lemma}{Lemma}
\newtheorem{theorem}{Theorem}
\newtheorem{corollary}{Corollary}

\newtheorem{exmp}{Example}

\newcommand{\ind}{\rotatebox[origin=c]{90}{$\models$}}

\newcommand{\X}{{\mathbf{X}}}
\newcommand{\Y}{{\mathbf{Y}}}
\newcommand{\Z}{{\mathbf{Z}}}
\newcommand{\G}{{\mathcal{G}}}

\newcommand{\cS}{{\mathcal{S}}}

\title{Vector Causal Inference between Two Groups of Variables}
\author{
    Jonas Wahl,\equalcontrib\textsuperscript{,\rm 1,\rm 2}
    Urmi Ninad,\equalcontrib\textsuperscript{,\rm 1,\rm 2}
    Jakob Runge\textsuperscript{\rm 1,\rm 2}
}
\affiliations{
    \textsuperscript{\rm 1}Technische Universit\"{a}t Berlin\\
    \textsuperscript{\rm 2} DLR Institut für Datenwissenschaften Jena \\
    wahl@tu-berlin.de, urmi.ninad@tu-berlin.de, runge@tu-berlin.de
}

\begin{document}

\maketitle
\begin{abstract}
\begin{quote}
Methods to identify cause-effect relationships currently mostly assume the variables to be scalar random variables. However, in many fields the objects of interest are vectors or groups of scalar variables.
We present a new constraint-based non-parametric approach for inferring the causal relationship between two vector-valued random variables from observational data. Our method employs sparsity estimates of directed and undirected graphs and is based on two new principles for groupwise causal reasoning that we justify theoretically in Pearl's graphical model-based causality framework. Our theoretical considerations are complemented by two new causal discovery algorithms for causal interactions between two random vectors which find the correct causal direction reliably in simulations even if interactions are nonlinear. We evaluate our methods empirically and compare them to other state-of-the-art techniques. 
\end{quote}
\end{abstract}
\let\thefootnote\relax\footnotetext{Accepted at AAAI 2023.}
\section{Introduction}
In recent years, many methods have been developed in order to infer cause-effect relationships between random variables from observational data, see e.g. \citet{PearlCausality,Spirtes2000,ShimizuLinGaM,PetJanSch17,runge2015identifying}. Most often these random variables are assumed to be scalar; an assumption which covers many but by no means all questions of interest in science. For instance, neuroscience researchers are often interested in the causal interactions between different brain \emph{regions} each of which are represented by a multitude of measurement locations in fMRI data. Similarly, in the Earth sciences, researchers would like to understand causal relationship between variables (such as sea surface temperature, air pressure or wind speed) that have each been measured at a number of grid locations in predefined areas on the planet. To make inferences in such a setup, standard approaches often proceed to drastically reduce the number of measurement variables, for instance by computing average values across each region of interest or by applying statistical dimension reduction techniques such as principal component analysis. However, aggregating data in such a way might lead to faulty causal conclusions: For instance, conditional independencies $\X \ind \Y | \Z$ between vectors might not be visible in their mean values \citep{Spirtes2000}. Opposing causal effects from different parts of a region might average to zero, and causally relevant information might be diluted or lost when considering aggregate variables, see Figure \ref{fig:dimred_pitfalls}. Furthermore, methods that rely on the non-Gaussianity of noise to make inferences about the causal direction, such as LinGaM \citep{ShimizuLinGaM}, are rendered weak by the averaging due to the central limit theorem driving the average noise closer to a Gaussian.

\begin{figure}
    \centering
    \includegraphics[scale=1.7]{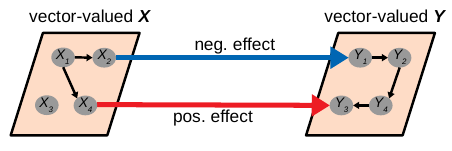}
    \caption{Two components in the vector-valued variable $\X$  have effects on $\Y$ that are of opposing signs, such that aggregation of $\X$ and $\Y$ leads to a dilution or cancellation of  dependence.}
    \label{fig:dimred_pitfalls}
\end{figure}

In this work, we aim to develop new techniques to infer the causal relationship between two groups of variables, represented as random \emph{vectors}, without a dimension reduction step. Assuming that the causal arrows go from one group to the other only, the most straightforward way to do so is to run a standard causal discovery algorithm, e.g. the PC algorithm, on all microvariables (i.e. all entries of both vectors), and then choose the cause group as the one that has the most edges pointing to the other group. When groups become large, this approach, henceforth called \emph{Vanilla-PC}, has disadvantages: since it needs to determine the full causal `microstructure', it has to run many conditional independence tests and due to the sequential error propagation of the PC algorithm becomes unreliable quickly at small sample size (see Subsection \ref{subsec.compPC} in the Supplement for an empirical illustration of this). In essence, Vanilla-PC computes more structure than is needed to answer the causal query at hand and therefore uses the data inefficiently. We therefore take a different road and combine the constraint-based approach for causal discovery with sparsity measures of the \emph{internal (causal) structure} of the groups.
Our methods are based on the following two principles for causal interaction between groups of random variables:
\begin{itemize}
\item[\textbf{(P1)}] generically, conditioning on the cause group does not \emph{create} new conditional dependencies \emph{within} the effect group;
\item[\textbf{(P2)}] generically, conditioning on the effect group does not \emph{delete} conditional dependencies \emph{within} the cause group.
\end{itemize}

\noindent Here, the term \emph{generically} is to be understood as in other assumptions for causal inference such as the causal Markov property, Faithfulness and the principle of independent cause and mechanism (ICM), see e.g. \citet{PetJanSch17}. That is to say, (P1) and (P2) can only be violated if the causal mechanism is in some way fine-tuned to the exogenous noise variables of the model.

\noindent We justify these principles more thoroughly by considering the setting where the scalar microvariables are modelled by a causal graphical model over a directed acyclic graph (DAG). In this setup, (P1) and (P2) turn out to be \emph{implied} by the causal Markov property and Faithfulness. Based on these principles, we prove that, in this purely causal setup, the correct causal direction is identifiable under weak assumptions (Theorem \ref{thm.main}). Moreover, we provide an algorithm for distinguishing cause from effect called \textbf{2G-VecCI.PC} (two  group vector causal inference, PC method) which is based on density estimation of each group through  the PC algorithm and which is sound and complete under said assumptions. To our knowledge, our method is the first non-parametric method to infer causal directionality between two groups of variables (other than Vanilla-PC). \\
\noindent In addition to the purely causal setting, we consider the setup where cause and effect group are related through a structural causal model of the form
\begin{align*}
\X \ &:= \ \eta_{\X}, \\
\Y \ &:= \ \mathbf{f}(\X, \eta_{\Y}), \qquad \eta_{\X} \ind \eta_{\Y}.
\end{align*}
Importantly, we do not assume that the individual components of the noise vector $\eta_{\X}$ (respectively $\eta_{\Y}$) are pairwise independent. Such a model can be reasonable when the internal interactions \emph{within} a variable group do not admit a straightforward causal interpretation while the interactions \emph{between} groups do. One might therefore consider such a model \emph{semi-causal}. For instance, if $\X$ describes a field of surface temperature measurements on different grid locations, stating that the measurement at location $i$ causes the measurement at location $j$ might be inappropriate. For such a semi-causal SCM, it is much harder to prove theoretical guarantees for the validity of (P1) and (P2), and we only demonstrate what violations of these principles would entail in a toy example. Nevertheless a second version of our inference algorithm dubbed \textbf{2G-VecCI.Full} (two group vector causal inference, full conditioning method) is able to find the correct causal direction in many cases in simulated data. In both the causal and the semi-causal setting, our algorithms make inferences by estimating the sparsity of graphs that encode conditional dependence or causal relationships within a variable group before and after conditioning on the other variable group.
At present our methods assume that samples are i.i.d., and that the sample size is larger than the total size of the vector-valued variables. Both methods are based purely on conditional (in)dependence relationships and, with appropriately chosen tests, work even if causal interactions are nonlinear.

We will present our theoretical identifiability results and the necessary assumptions in Section 3\ref{sec.theory}. After that, we describe two different versions of our algorithm for causal discovery between variable groups in Section 4\ref{sec.algorithms}. In Section 5, 
we analyse the empirical performance of these algorithms in experiments with synthetic data and compare it to that of other approaches (Vanilla-PC and the Trace Method \citep{Janzing09, ZscheiJanZhang12}). We also consider a real world climate science example of surface temperatures in the El Ni\~{n}o Southern Oscillation (ENSO 3.4) region in the pacific and in British Columbia to test our algorithms. We conclude with a discussion and outlook in Section 6. 

\section{Related Work}
Although the majority of causal discovery results focus on scalar variables, the idea to study causal interactions between groups of random variables is not new. On the theoretical side,  \citet{Rubensteinetal17}, \citet{ChaEbPer16}, \citet{ChalupkaEtAl16} and  \citet{ChaEbPer17} discuss to which extent micro variables can be aggregated to macro variables without losing causal information. \citet{ParKas17} discuss theoretical assumptions of multi-group causal discovery in connection to the PC-algorithm (as opposed to the two group identifiability problem discussed here). For two linearly interacting groups, \citet{Janzing09} introduce a causal discovery algorithm called the Trace method, see also \citet{ZscheiJanZhang12}. In \citet{EntHoy12}, scalar causal discovery techniques based on non-Gaussianity assumptions such as LiNGaM \citep{ShimizuLinGaM} are generalized to the vector-valued setting. We summarize the existing approaches, as well as their assumptions, strengths and weaknesses in Table \ref{table.comparison} in the supplement.

\section{Identifiability Results} \label{sec.theory}

\subsubsection{Theoretical Setup}

We will consider scalar random variables $X_1,\dots,X_n,Y_1,\dots,Y_m$ that are grouped into two vectors $\X = \{ X_1,\dots,X_n \}, \Y = \{ Y_1,\dots,Y_m \}$. We assume that the data is generated by a causal process $\X \rightarrow \Y$ as outlined below, and our goal is to infer the correct causal direction from the observational distribution $P_{\X,\Y}$. We will operate under different sets of assumptions that relate $P_{\X,\Y}$ to causal representations (see Model 1 and 2 below). We refer to the mathematical appendix for a quick overview on directed and undirected graphs, d-separation, the causal Markov property, Faithfulness and causal sufficiency. We will always assume that a statistical association $\X \centernot\ind \Y$ is present in the data.

\paragraph{Model 1 (Unidirectional Causal Vector Model)}
The scalar variables $X_1,\dots,X_n,Y_1,\dots,Y_m$ are represented as the nodes of a directed acyclic graph (DAG) $\G$. In addition, we assume that 
    \begin{itemize}
         \item[\textbf{(A1)}] The joint distribution $P_{\X,\Y}$ fulfills the causal Markov property and is faithful to $\G$. In other words, d-separation in $\G$ completely characterizes the conditional independencies present in $P_{\X,\Y}$. This implies that the model is causally sufficient, i.e. no hidden confounders are present. 
         \item[\textbf{(A2)}] Arrows between groups only point in one direction, i.e., without loss of generality, the $\X \rightarrow \Y$-direction. In other words, there can be no arrow $X_k \leftarrow Y_{\ell}$ for any $X_k \in \X, Y_{\ell} \in \Y$. 
    \end{itemize}
    


\noindent We will now justify principles (P1) and (P2) in this scenario. The main result of this section will be summarized in Theorem \ref{thm.main}. Proofs of the following results will be given in the technical appendix. 

\begin{lemma} \label{thm.principle1}
Assume that the assumptions of Model 1 are satisfied. Then principle (P1) holds in the sense that there is no subset $\cS \subset \Y$ of the effect group and nodes $Y_k, Y_{\ell} \in \Y$ such that
\begin{align*}
    Y_k \ind Y_{\ell} \ | \ \cS \qquad \text{and} \qquad Y_k \centernot \ind Y_{\ell} \ | \ \cS, \X.
\end{align*}
\end{lemma}

\noindent Next, we say that a causal vector model $\X \to \Y$ satisfies condition
\begin{itemize}
    \item[\textbf{(C1)}] if there is a subset $\cS \subset \X$ and scalar variables $X_i,X_j \in \X$ such that
\begin{align*}
    X_i \ind X_j \ | \ \cS \qquad \text{and} \qquad X_i \centernot \ind X_j \ | \ \cS, \Y. 
\end{align*}
\end{itemize}

In the appendix, we will characterize (C1) graphically and depict some motivational examples. For instance, Condition (1) is satisfied only if there exists a \emph{cross-regional v-structure} $X_i \rightarrow Y_k \leftarrow X_j$. We can now deduce the following result for cause-effect identification. It states that whenever conditioning on a group creates dependencies within the other group, by (P1) the former must be the effect and the latter must be the cause.

\begin{corollary} \label{thm.ident1}
Assume that the assumptions of Model 1 as well as (C1) are satisfied. Then, the causal direction $\X \rightarrow \Y$ can be inferred from the observational distribution $P_{\X,\Y}$.  
\end{corollary}

\noindent Next, we justify principle (P2).

\begin{lemma} \label{thm.principle2}
Assume that the assumptions of Model 1 are satisfied. Then principle (P2) holds, in the sense that there is no subset $\cS \subset \X$ of the cause group and nodes $X_i, X_{j} \in \X$ such that
\begin{align*}
    X_i \centernot\ind X_{j} \ | \ \cS \qquad \text{and} \qquad X_i \ind X_{j} \ | \ \cS, \Y.
\end{align*}
\end{lemma}

\noindent Again, we need to ensure that conditioning on the cause vector \emph{does} delete dependencies within the effect vector. We therefore say that a causal vector model $\X \to \Y$ satisfies
\begin{itemize}
    \item[\textbf{(C2)}] if there is a subset $\cS \subset \Y$ and scalar variables $Y_k,Y_{\ell} \in \Y$ such that
\begin{align*}
    Y_k \centernot\ind Y_{\ell} \ | \ \cS \qquad \text{and} \qquad Y_k \ind Y_{\ell} \ | \ \cS, \X. 
\end{align*}
\end{itemize}

\noindent For example, condition (C2) is satisfied if $Y_k, Y_{\ell}$ can be d-separated by $\cS \subset \Y$ \emph{in the subgraph over} $\Y$ and there is a common confounder $Y_k \leftarrow X_i \rightarrow Y_{\ell}$. Again, we will provide a full graphical characterization of (C2) and some examples in the appendix.

\begin{corollary} \label{thm.ident2}
Assume that the assumptions of Model 1 as well as (C2) are satisfied. Then, the causal direction $\X \rightarrow \Y$ can be inferred from the observational distribution $P_{\X,\Y}$. 
\end{corollary}

\noindent We summarize the results above in the following theorem.

\begin{theorem} \label{thm.main}
Assume that the assumptions of Model 1 are satisfied and that at least one of the conditions (C1) or (C2) holds. Then. the causal direction $\X \rightarrow \Y$ can be inferred from the observational distribution $P_{\X,\Y}$.
\end{theorem}

\paragraph{Model 2 (Unidirectional Semi-Causal Vector Model)}
Model 2 assumes that the variables $X_1,\dots,X_n,Y_1,\dots,Y_m$ are generated by the \emph{semi-causal} structural causal model
\begin{align*}
\X \ &:= \ \eta_{\X}, \\
\Y \ &:= \ \mathbf{f}(\X,\eta_{\Y}), \qquad \eta_{\X} \ind \eta_{\Y},
\end{align*}
where as mentioned before, we do \emph{not} assume that the components within the noise terms $\eta_{\X}, \eta_{\Y}$ are pairwise independent. We can encode the conditional independencies within the $\X$-group graphically by drawing an undirected edge $X_i -X_j$ if and only if 
\begin{align*}
X_i\centernot\ind X_j \ | \ \X\backslash \{X_i\, X_j\}
\end{align*}   
and similarly within the $\Y$-group by drawing an undirected edge $Y_k-Y_{\ell}$ if and only if 
\begin{align*}
\eta_k \centernot\ind \eta_{\ell} \ | \ \eta \backslash \{\eta_k , \eta_{\ell}\}.
\end{align*} 
The undirected graphs obtained in this way will be denoted by $\G'_{\X}, \G'_{\Y}$ respectively.

In this way, we have encoded the distributions of $\X,\Y$ as \emph{Markov random fields} over undirected graphs. Markov random fields of this kind are sometimes employed to model spatial or spatio-temporal measurements on grids \citep{Song_et_al2008} such as for instance surface temperature measurements \citep{Vaccaro_et_al21}.

In this setup, it is harder to find exact conditions that formally imply principles (P1), (P2). We lack a non-finetuning statement that is as general as faithfulness in the causal setting and quantifying such a statement would probably require additional assumptions on the functional form of the model or the noise distribution.

Let us illustrate why it is nevertheless reasonable to accept (P1) and (P2) with the following toy example:
\begin{exmp}
Consider the model 
\begin{align*}
\begin{pmatrix}Y_1\\Y_2\end{pmatrix} = \begin{pmatrix}b X_1 + \eta_1 \\c X_2 + \eta_2\end{pmatrix},
\end{align*}
where $(X_1,X_2)$ are jointly normal with mean $\mathbb{E}[\X]= \mathbf{0}$, $\mathrm{Var}(X_1) = \mathrm{Var}(X_2) =1 $ and $\mathrm{Cov}(X_1,X_2) = a$. The error terms $(\eta_1,\eta_2)$ are jointly normal with mean $\mathbb{E}[\eta]=\mathbf{0}, \mathrm{Var}(\eta_1) = \mathrm{Var}(\eta_2) =1$ and $\mathrm{Cov}(\eta_1,\eta_2) = d$. The only way conditioning on $\X$ could create a dependency in $\Y$ would be if $Y_1 \ind Y_2$ in the first place, which is equivalent to $abc = -d$. Thus conditioning on the cause can only create dependencies out of independencies that arose from a finetuning of the coefficients of the mechanism $b,c$ to the coefficients of the noise terms $a,d$.

Similarly, for (P2) to be violated in this example, conditioning on $\Y$ would have to delete the dependency of $X_1$ and $X_2$. This would entail another (more involved) algebraic equation that the coefficients would have to satisfy. Thus coefficients describing the causal mechanism could not be chosen independently of the noise terms. 
\end{exmp}

\subsubsection{Graph edge density criterion for identifying causal direction}

A practical way to make use of Theorem \ref{thm.main} is to compare the edge densities of the internal graphs of one group before and after conditioning on the other group. Say we are interested in these internal graphs of the vector $\X$ in Model 1 where internal graphs are formalized as DAGs. In this case we then compare the number of edges in the (skeleton of the) CPDAG $\G_{\X}$ that is the output of the PC-algorithm run on the scalar variables in $\X$ to the number of edges in the CPDAG $\G_{\X| \Y}$ that is the output of the PC-algorithm over $\X$ in which $\Y$ is added as a conditioning set to every independence test. We normalize these edge counts by the maximal possible number of edges $\mathrm{edgeMax} = n(n-1)/2$ to obtain edge densities
\begin{align*}
\mathrm{edgeDens}(\G_{\X}) &= \frac{\text{number of edges of } \G_{\X}}{\mathrm{edgeMax}} \\ \mathrm{edgeDens}(\G_{\X|\Y}) &= \frac{\text{number of edges of } \G_{\X|\Y}}{\mathrm{edgeMax}}.
\end{align*}
The edge densities $\mathrm{edgeDens}(\G_{\Y}), \ \mathrm{edgeDens}(\G_{\Y|\X})$ are defined analogously. Presuming that the PC-algorithm is run with perfect oracle independence tests, Conditions (C1) and (C2) then have the following implications.

\begin{theorem} \label{thm.edgedensity1}
As before we assume the causal direction to be $\X \rightarrow \Y$ and that the assumptions of Model 1 are satisfied. If moreover condition (C1) holds, then
\begin{align}\label{eq.d_1}
d(\X|\Y) := \mathrm{edgeDens}(\G_{\X|\Y}) - \mathrm{edgeDens}(\G_{\X}) > 0.
\end{align}
If condition (C2) holds, then
\begin{align}\label{eq.d_2}
d(\Y|\X):= \mathrm{edgeDens}(\G_{\Y|\X}) - \mathrm{edgeDens}(\G_{\Y}) < 0.
\end{align}
In either case, we have $d(\X|\Y) > d(\Y|\X)$ and the causal direction can thus be inferred from the sign of $d(\X|\Y) - d(\Y|\X)$.
\end{theorem}
As a consequence, when the causal direction is unknown, we can infer it from the observational distribution  by computing $d(\X|\Y)$ and $d(\Y|\X)$ and choosing $\X$ as the cause if the former is larger and $\Y$ if the latter is larger. Note that this approach also works when $\X$ and $\Y$ have different internal edge densities.

\noindent In Model 2 we have to replace $\G_{\X}$ by the undirected conditional independence graph  $\G'_{\X}$ and $\G_{\X|\Y}$ by the graph $\G'_{\X|\Y}$ that has edges $X_i - X_j$ iff 
\[X_i \centernot\ind X_j | \X \backslash \{X_i,X_j \}, \Y. \] 

If we now replace $d(\X | \Y), d(\Y | \X)$ by

\begin{align}\label{eq.d_hat_1}
d'(\X|\Y) := \mathrm{edgeDens}(\G'_{\X|\Y}) - \mathrm{edgeDens}(\G'_{\X}),
\end{align}
and
\begin{align}\label{eq.d_hat_2}
d'(\Y|\X):= \mathrm{edgeDens}(\G'_{\Y|\X}) - \mathrm{edgeDens}(\G'_{\Y}),
\end{align}

we still observe empirically (see below) that the sign of $d'(\X|\Y) - d'(\Y|\X)$ is able to read the causal direction from the observational distribution quite efficiently.
 

\section{Algorithms for Cause-Effect Identification} \label{sec.algorithms}
We now present two algorithms for cause-effect identification that are tailored to the two different models. 2G-VecCI.PC is particularly useful if the user believes the assumption of Model 1 to be valid. Recall that the PC-algorithm is an algorithm for causal discovery on scalar variables that consists of a \textbf{skeleton phase} to identify the skeleton of the causal graph and an orientation phase to determine causal directions. For the computation of edge densities, 2G-VecCI.PC will use the PC-algorithm's skeleton phase.
\RestyleAlgo{ruled}
\begin{algorithm}[!h]
\LinesNumbered
\caption{2G-VecCI.PC}\label{pseudocode1}
\KwData{two arrays containing samples of $\X$ and $\Y$, parameter $\alpha \in [0,1]$.}
\KwResult{variable with values '$\X$ is the cause of $\Y$', or '$\Y$ is the cause of $\X$', or 'Causal direction cannot be determined'.}
Run \textbf{skeleton phase} on the components of $\X$\;
Compute $\widehat{\mathrm{edgeDens}(\G_{\X})} = \frac{\text{number of edges found in skeleton phase }}{\mathrm{edgeMax}}$\; 
Run \textbf{skeleton phase} on the components of $\X$\ with $\Y$ added to the conditioning set of every independence test\;
Compute $\widehat{\mathrm{edgeDens}(\G_{\X|\Y})}$ and
$\widehat{d(\X|\Y)} = \widehat{\mathrm{edgeDens}(\G_{\X|\Y})} - \widehat{\mathrm{edgeDens}(\G_{\X})}$\;
Repeat (1) to (4) with exchanged roles of $\X,\Y$ to get $\widehat{d(\Y|\X)}$\;
Compute $\mathrm{Crit} = \widehat{d(\X|\Y)} - \widehat{d(\Y|\X)}$\;
\lIf{ $|\mathrm{Crit}| < \alpha$}{return \\ 'Causal direction cannot be determined'}
\lIf{ $\mathrm{Crit} > \alpha$}{return '$\X$ is the cause of $\Y$'}
\lIf{ $\mathrm{Crit} < -\alpha$}{return '$\Y$ is the cause of $\X$'}
\end{algorithm}

The parameter $\alpha \in [0,1]$ is a sensitivity parameter that controls the algorithm's agnosticism. If $\alpha$ is chosen large, then it will return a definite result only in clear cut cases. 

If 2G-VecCI.PC is run with a consistent independence test, i.e. with one that recovers conditional independence relations perfectly in the infinite sample limit, it consistently estimates $d(\X|\Y)$ and $d(\Y|\X)$. Note that this also applies to nonlinear relations, our approach does not rely on functional assumptions. Therefore, the algorithms consistency follows directly from Theorem \ref{thm.edgedensity1}.

\begin{corollary} \label{cor.consistency}
If the assumptions of Model 1 are satisfied and if at least one of (C1) or (C2) holds, then 2G-VecCI.PC run with a consistent CI test returns the correct causal direction in the infinite sample limit.
\end{corollary}

The above result is fully non-parametric, but in practice, under appropriate assumptions, one way to include $\Y$ as a conditioning set for independence tests within $\X$ for computing $d(\X|\Y)$, is to regress $\Y$ on $\X$ and run skeleton phase on the residuals of this regression, and vice-versa for $d(\Y|\X)$. In particular, if the relationship between groups is assumed linear, this is our method of choice. As an alternative to 2G-VecCI.PC, it is also possible to run a 'one sided' version that only computes $\widehat{d(\X|\Y)}$ (or $\widehat{d(\Y|\X)}$) and decides solely based on its sign. This version of the algorithm is computationally less costly but we have found it to be more frail to statistical errors than the full version in practice.  \\

\noindent Our second algorithm 2G-VecCI.Full differs from 2G-VecCI.PC in that skeleton phase of the PC algorithm is replaced by independence tests 
\begin{align*}
 X_i \ind X_j \ &| \ \X \backslash \{X_i,X_j \}, \\
  X_i \ind X_j \ &| \ \X \backslash \{X_i,X_j \},\Y, \quad \forall{i,j},
\end{align*} 
to compute edge densities for $d'(\X|\Y)$ since we are interested in the (conditional) dependence graph. Again, it is reasonable to regress $\X$ on $\Y$ in many practical settings and to perform independence tests on the residuals. We proceed analogously with $\X$ and $\Y$ exchanged to compute $d'(\Y|\X)$.
 
\noindent Another remark of practical importance is that 2G-VecCI.Full is also applicable in the context of Model 1, although it measures the edge densities of the \emph{moralized graph} of the causal DAGs, rather than the edge densities of the DAGs itself (see the appendix for the definition of the moralized graph). In this sense, 2G-VecCI.Full is more general than 2G-VecCI.PC.  However, if the causal DAG over a group differs dramatically from its moralized graph, 2G-VecCI.PC should be preferred. As an extreme example, consider the case where there is $X_i \in \X$ that is a common child of all other elements of $\X$ and similarly $Y_j \in \Y$ that is a common child of all other elements of $\Y$. In this situation, the moralized graph over each group would be fully connected (with or without conditioning) and 2G-VecCI.Full would not be suitable.

\RestyleAlgo{ruled}
\begin{algorithm}[!h]
\LinesNumbered
\caption{2G-VecCI.Full}\label{pseudocode2}
\KwData{two arrays containing samples of $\X$ and $\Y$, parameter $\alpha \in [0,1]$.}
\KwResult{variable with values '$\X$ is the cause of $\Y$', or '$\Y$ is the cause of $\X$', or 'Causal direction cannot be determined'.}
\For{$X_i \neq X_j \in \X$}{test $X_i \ind X_j \ | \ \X \backslash \{X_i,X_j \}$ \;
\lIf{dependent}{$\widehat{\mathrm{edgeCount}(\G_{\X}')}+=1$}
test $X_i \ind X_j \ | \ \X \backslash \{X_i,X_j \}, \Y$ \;
\lIf{dependent}{$\widehat{\mathrm{edgeCount}(\G_{\X|\Y}')}+=1$}}
Compute $\widehat{\mathrm{edgeDens}(\G_{\X}')} = \frac{\widehat{\mathrm{edgeCount}(\G_{\X}')}}{\mathrm{edgeMax}}$\; 
Compute $\widehat{\mathrm{edgeDens}(\G_{\X|\Y}')} = \frac{\widehat{\mathrm{edgeCount}(\G_{\X|\Y}')}}{\mathrm{edgeMax}}$\; 
Compute
$\widehat{d'(\X|\Y)} = \widehat{\mathrm{edgeDens}(\G_{\X|\Y}')} - \widehat{\mathrm{edgeDens}(\G_{\X}')}$\;
Repeat 1-9 with $\X,\Y$ exchanged to get $\widehat{d'(\Y|\X)}$\;
Compute $\mathrm{Crit} = \widehat{d'(\X|\Y)} - \widehat{d'(\Y|\X)}$\;
Repeat steps 7-10 of 2G-VecCI.PC.

\end{algorithm}
\section{Experimental Results} \label{sec.experiments}

\subsection{Simulated Data} \label{subsec.simulated}

We depict several results for 2G-VecCI.PC and 2G-VecCI.Full for linear and nonlinear (quadratic) models in the main text. For linear models we generate the empirical distributions of $\X$ and $\eta_{\Y}$ by linear SCMs with randomly chosen coefficients and Gaussian noises with randomized variances in the range $[0.5,2.0]$. We then generate a random $n \times m$ interaction matrix $A$ and set $\Y = A \X + \eta_{\Y}$. Models vary along the following parameters:

\begin{itemize}
\item sample size (between $50$ and $500$);
\item group sizes $n$ and $m$ (between $3$ and $100$);
\item edge densities within $\X$ and $\eta_Y$ (between $1\%$ and $90\%$ of all possible edges);
\item density of the interaction matrix $A$ (between $1\%$ and $90\%$ of all possible entries non-zero);
\item effect size, i.e. size of the entries in $A$ (uniformly randomly drawn from different intervals).
\end{itemize} 

For each parameter choice, 100 random models are generated. In both algorithms, we test for conditional independencies using the partial correlation test at significance level $\tilde{\alpha} = 0.01$. We choose the sensitivity parameter to be $\alpha = 0.01$ if not specified differently. We also compare our algorithms to existing approaches, see below. Computations were done on BullSequana XH2000 with AMD 7763 CPUs. We observe the following:

\begin{itemize}
\item Generally speaking, 2G-VecCI.Full outperforms 2G-VecCI.PC on our simulated data even though the data is generated by a causal model.
\item Performance increases with increasing effect size and increasing interaction density and decreases with increasing edge densities within variable groups.
\item More precisely, both algorithms perform best when dependencies within both variable groups are sparse to medium sparse ($<10\%$ of all possible edges). Nevertheless, for high sample sizes (e.g. $500$), the correct causal direction is still inferred reliably for large groups ($100$ variables per group) even if variable groups are quite dense ($30\%$ of all possible edges).
\item For the algorithms to perform well, the interaction matrix should not be too sparse (i.e. $<10\%$ of non-zero entries) and effect sizes should not be too small $<0.1$.This type of 'weak mechanism problem' is typical for constraint-based causal inference algorithms in general as one operates close to the non-faithful regime.
\end{itemize}
\subsubsection{Complexity of 2G-VecCI.PC and 2G-VecCI.Full}

As in other PC-based methods, the number of CI tests run by 2G-VecCI.PC is data dependent and may increase exponentially in the worst case, i.e. for high group and interaction densities where separating sets need to be large. 2G-VecCI.Full on the other hand runs $2(n^2+m^2)$ CI tests where $n,m$ are the group sizes, independent of the specifics of the data (but with large conditioning sets). Therefore, if groups and interactions are assumed very sparse, 2G-VecCI.PC may be less costly than 2G-VecCI.Full while 2G-VecCI.Full should be preferred over 2G-VecCI.PC when groups are assumed to be reasonably dense and in practice, we have typically found 2G-VecCI.Full to perform significantly faster.  For an in-depth discussion on computational cost of the PC-algorithm, we refer the reader to \citet{KaBue07} and \citet{LeFastPC}.

To test our methods for nonlinear interactions, ground truth data is generated as in the linear case except that we use the model $\Y = A \X^2 + \eta_{\Y}$, where $\X^2$ is shorthand for the vector of squared entries of $\X$. 

Here, our methods are affected more strongly by increasing group sizes as nonlinear CI-tests tend to be much slower than tests for partial correlation.
Nevertheless, 2G-VecCI.Full run with the Gaussian Process distance correlation independence test still finds the correct causal direction significantly better than a random choice would for groups of size 15 (with 100 samples) and 25 (with 200 samples), see Figure 4. At present, we did not implement non-linear interactions for 2G-VecCI.PC. For large groups, performance could potentially be sped up by reducing dimensions locally. For instance, if the data is structured spatially, small subregions could be averaged to scalar variables to obtain a coarser variable group. The approach of \citet{ChaEbPer17} might be helpful here and combining it with our work might be an avenue for future research.

\subsection{Comparison to other methods} \label{sec.limitations}

Two established methods for inferring causal relations between variable groups are multivariate LiNGaM \citep{ShimizuLinGaM} \citep{EntHoy12} and the Trace Method \citep{Janzing09} \citep{ZscheiJanZhang12}. In contrast to our methods, both of these techniques assume interactions to be linear and LiNGaM additionally requires non-Gaussian noise to be applicable. In simulations with linearly interacting groups and Gaussian noise, the trace method and both versions of 2G-VecCI perform comparably well, see Figures \ref{fig.groupsize30}, \ref{fig.different_group_sizes_low_group_density} and \ref{fig.groupsize100}, although the trace method is significantly faster. We also analysed the performance of our method against a baseline PC algorithm (Vanilla-PC) by treating each component in the vector-valued variables as a separate node and counting the arrow directions from (nodes belonging to) one group to the other. In general our methods outperform Vanilla PC except when groups are small and the interaction matrix is sparse, see Figures \ref{fig.groupsize30} and \ref{fig.different_group_sizes_low_group_density} as well as Section \ref{subsec.compPC} and Figures \ref{fig.different_group_sizes_low_interaction_density},  \ref{fig.different_sample_sizes}  in the supplement.

\subsection{A real-world example} \label{subsec.realworld}
In order to test our algorithms in a typical causal discovery setting in Earth sciences, we consider surface temperatures over the ENSO 3.4 region and over British Columbia (denoted by BCT) from 1948-2021. We consider this example because a causal effect of temperatures in the tropical pacific on those in North America is established in climate science. Additionally, \citet{Runge17Science} used this example to test the PCMCI algorithm and found a causal link $\text{Nino}_{t-2} \to \text{BCT}_{t}$, i.e. at a time-lag of two months. 

The data is first de-seasonalized and any long-term trend is removed from the raw time series with a Gaussian kernel smoothing mean with a bandwidth of $\sigma$ = 120 months as in \citet{Runge17Science}. Furthermore, in order to mitigate auto-correlation in time, we consider means of ENSO temperature anomalies from October to December and means of BCT anomalies from January to March, leading to 73 samples each for the two groups $\X$ and $\Y$. To get more robust results, we consider different coarse grainings of the two regions, for instance every second grid-box for ENSO and every third grid-box for BCT anomalies. This has the additional effect of reducing group sizes in comparison to the sample size. We moreover reject coarse grainings that correspond to a difference of more than 10 grid-boxes between the two groups, in order to avoid a bias due to region sizes. Algorithm 2G-VecCI.Full with the partial correlation CI-test then computes $\mathrm{Crit} = \widehat{d'(\X|\Y)} - \widehat{d'(\Y|\X)}$, see \eqref{eq.d_hat_1} and \eqref{eq.d_hat_2}. We find the mean and the standard deviation of the $\mathrm{Crit}$ values to be $\mu = 0.031 \text{ and } \sigma = 0.063$, respectively, indicating a causal effect of ENSO on BCT. If the sensitivity parameter $\alpha$ is chosen  to be $0.01$, then the fraction of correct inferences is $0.59$ and the fraction of wrong inferences is $0.27$. Thus 2G-VecCI.Full deduces the correct casual direction Nino $\to$ BCT with high probability. Algorithm 2G-VecCI.PC with partial correlation and $\alpha = 0.01$ on the other hand yields $\mu = 0.001 \text{ and } \sigma = 0.019$, respectively and is thus indecisive. We attribute the low detection power of 2G-VecCI.PC to the reduced effect size arising from unobserved confounders and insufficiently mitigated auto-correlation in this simplified study.
NCEP-NCAR Reanalysis 1 data was provided by NOAA PSL, Boulder, Colorado, USA, from their website at https://psl.noaa.gov, see \citet{TheNCEPNCAR40YearReanalysisProject}.

\section{Discussion and Outlook}
We have introduced two new algorithms to infer causal direction between two potentially high-dimensional groups of variables and provided a theoretical analysis of groupwise causal inference in the DAG-based causality framework. 

The \textbf{main strengths} of our work are that it contains a novel identifiability result for the unidirectional causal vector model and practical implementations using density estimates of the vector-valued variables. It is comparable to the trace method in the high sample regime when interactions are linear and better when the group sizes are large and interactions are sparse. Moreover our methods are also able to deal with non-linear interactions.  
Currently, the \textbf{main weaknesses} are that we work with an i.i.d.~assumption on the data samples and that unobserved confounding variables are not addressed. Furthermore, our algorithms have slower runtime than the trace method or dimension reduction techniques.

In \textbf{future work}, we plan to extend this work to the setting of multiple variable groups similar to \citet{ParKas17} as well as to use partial dimension reduction techniques. We also plan to relax the i.i.d. assumption to better deal with autocorrelations following \citet{Runge17Science}.


\begin{figure*}[!ht] 
    \centering
    \begin{subfigure}[t]{0.22\textwidth}
        \centering
        \includegraphics[scale=0.3]{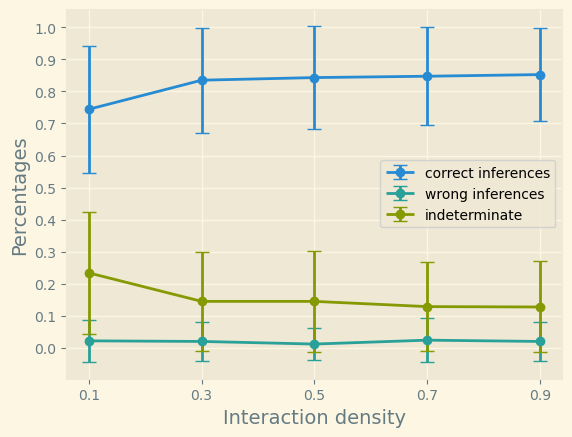}
        \caption{2G-VecCI.PC}
    \end{subfigure}%
    ~ 
    \begin{subfigure}[t]{0.22\textwidth}
        \centering
        \includegraphics[scale=0.3]{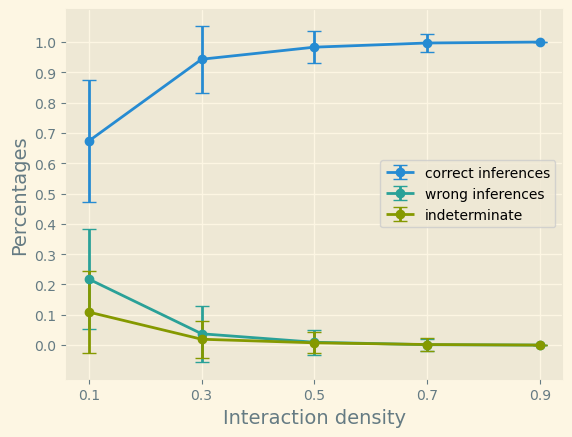}
        \caption{2G-VecCI.Full}
    \end{subfigure}
~
    \begin{subfigure}[t]{0.22\textwidth}
        \centering
        \includegraphics[scale=0.3]{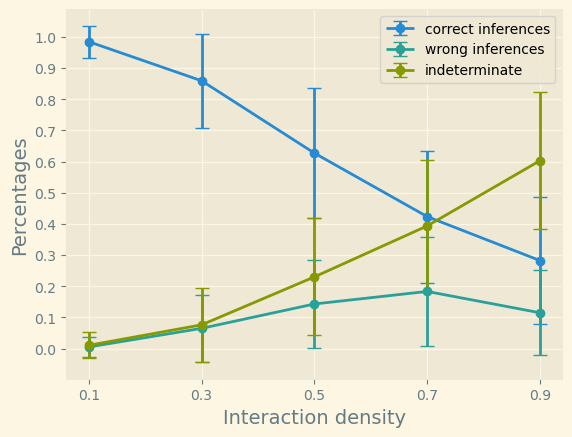}
        \caption{Vanilla-PC}
    \end{subfigure}
    ~ 
    \begin{subfigure}[t]{0.22\textwidth}
        \centering
        \includegraphics[scale=0.3]{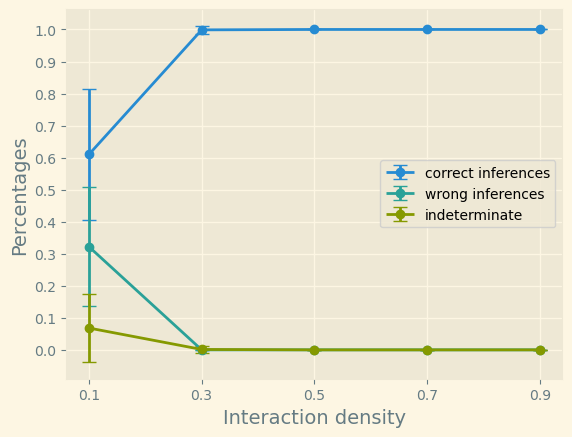}
        \caption{Trace Method}
    \end{subfigure}
    \caption{Performance of 2G-VecCI.PC, 2G-VecCI.Full, Vanilla-PC and the trace method for groups of size 30 with 100 available samples and linear interactions. Performance is shown along increasing density of the interaction matrix $A$, and percentages are averaged across different parameters for internal group densities ($1\%, 5\%, 10\%, 30\%$ of all possible edges present). Non-zero entries of $A$ are drawn uniformly randomly from $[-0.7,0.7]$ and 100 random models are run per parameter combination. All approaches except Vanilla PC recover the correct causal direction well. Vanilla-PC is challenged by dense interaction matrices as this increases the overall density of the causal graph over all microvariables. We set the sensitivity parameter of Vanilla-PC to $10^{-4}$ to ensure that the lack of performance is not due to an overly conservative choice.} \label{fig.groupsize30}
\end{figure*}

\begin{figure*}[!ht] 
    \centering
    \begin{subfigure}[t]{0.22\textwidth}
        \centering
        \includegraphics[scale=0.3]{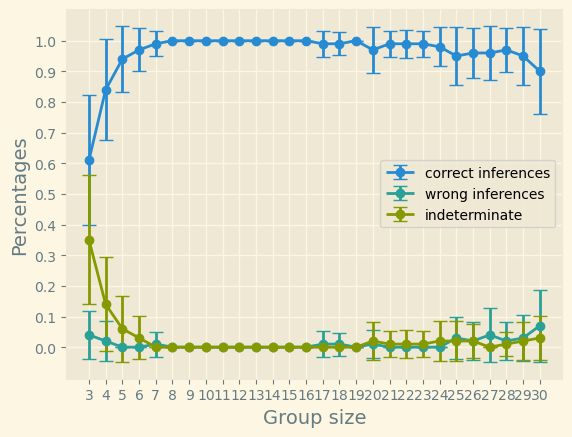}
        \caption{2G-VecCI.PC}
    \end{subfigure}%
    ~ 
    \begin{subfigure}[t]{0.22\textwidth}
        \centering
        \includegraphics[scale=0.3]{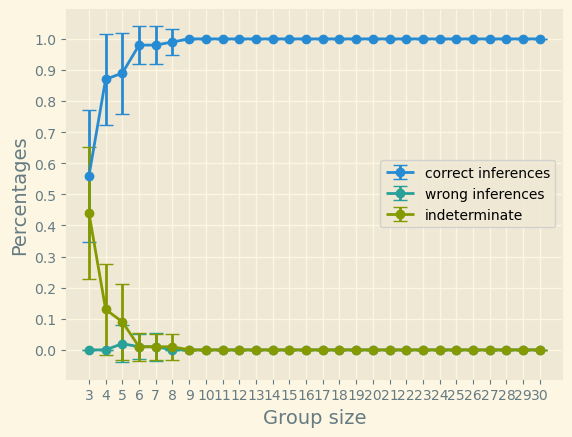}
        \caption{2G-VecCI.Full}
    \end{subfigure}
~
    \begin{subfigure}[t]{0.22\textwidth}
        \centering
        \includegraphics[scale=0.3]{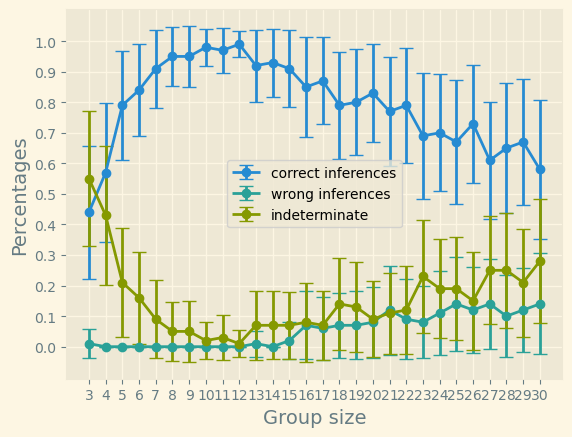}
        \caption{Vanilla-PC}
    \end{subfigure}
    \begin{subfigure}[t]{0.22\textwidth}
        \centering
        \includegraphics[scale=0.3]{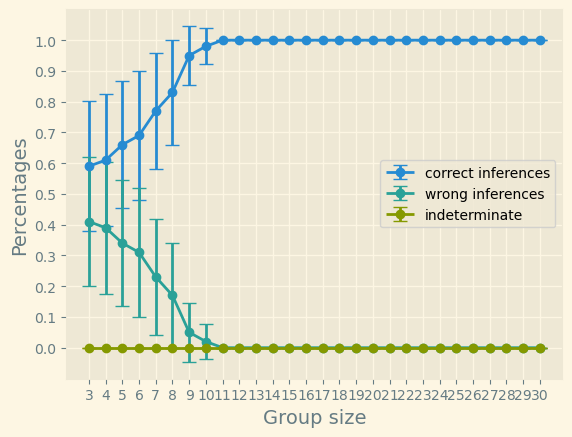}
        \caption{Trace Method}
    \end{subfigure}
    \caption{Performance of 2G-VecCI.PC, 2G-VecCI.Full, Vanilla-PC and the Trace Method for groups different sizes and low densities ($10 \%$). The density of the interaction matrix $A$ are set to $50\%$. Non-zero entries of $A$ are drawn uniformly randomly from $[-0.7,0.7]$ and 100 random models are run per parameter combination with 100 samples each. We set the sensitivity parameter of both 2G-VecCI.PC and Vanilla-PC to $10^{-4}$. Vanilla PC performs well for small groups but decreases in performance as group sizes grow.} \label{fig.different_group_sizes_low_group_density}
\end{figure*}

\begin{figure*}[!ht] 
    \centering
    \begin{subfigure}[t]{0.22\textwidth} 
        \centering
        \includegraphics[scale=0.3]{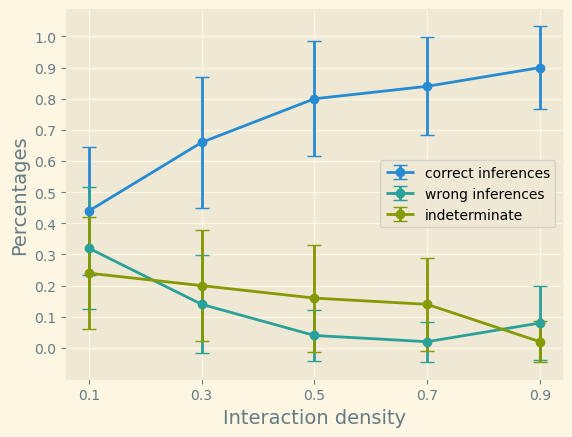}
        \caption{Nonlinear, group size 15}
    \end{subfigure}
    ~ 
    \begin{subfigure}[t]{0.22\textwidth}
        \centering
        \includegraphics[scale=0.3]{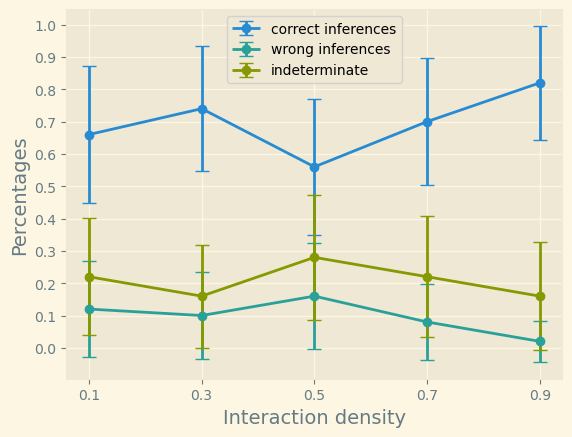}
        \caption{Nonlinear, group size 25}
    \end{subfigure}
    \caption{Performance of 2G-VecCI.Full for quadratic interactions of size 15 (left) and 25 (right) with 200 samples. Performance is shown along increasing density of the interaction matrix $A$. Non-zero entries of $A$ are drawn uniformly randomly from $[-0.7,0.7]$, groups are assumed medium dense on the left ($\approx 5$ connections per variable) and sparse on the right ($\approx 3$ connections per variable). 50 random models are run per choice of parameters. In both cases 2G-VecCI.Full finds the correct causal direction better than chance would except when $A$ is extremely sparse.} \label{fig.nonlinear}
\end{figure*}

\clearpage
\clearpage

\section{Acknowledgements}
This research has been supported by grant no. 948112 \emph{Causal Earth} of the European Research Council (ERC). This work used resources of the Deutsches Klimarechenzentrum (DKRZ) granted by its Scientific Steering Committee (WLA) under project ID bd1083. We thank the reviewers for their insightful remarks and suggestions.

\bibliography{References.bib}
 \clearpage
\clearpage

\section{Vector causal inference between two groups of variables: Technical Appendix}

\subsection{Graphs and separations}

We will first recall the notion of $d$-separation on a \emph{directed acyclic graph (DAG)} $\G = (V,E)$. If $\G$ contains an edge $A \rightarrow B$, we call $A$ a \emph{parent} of $B$ and $B$ a \emph{child} of $A$. Similarly if there is a directed path $A \rightarrow \dots \rightarrow B$, we call $B$ a \emph{descendant} of $A$. An undirected path between nodes $A$ and $B$ is a sequence of edges $A - \dots - \dots - B$ where the directionality of the edges is ignored. A node $C \in V \backslash \{ A,B \}$ on such a path is called a \emph{collider} if both adjacent edges point towards $C$, i.e. $\rightarrow C \leftarrow$, otherwise it is called a \emph{non-collider}. A path $p$ is said to be blocked by a (possibly empty) set of nodes $\cS$ if $\cS$ contains a non-collider $S \in \cS$ on $p$ or if there exists a collider on $p$ such that none of its descendants is contained in $\cS$. Finally, we say that $A$ and $B$ are $d$-separated by $\cS$ if every path between them is blocked by $\cS$. 

If nodes on a DAG $\G$ correspond to scalar random variables $X_1,\dots X_n,$ with joint distribution $P_{X_1,\dots,X_n}$, we say that $P_{X_1,\dots,X_n}$ has the (causal) \emph{Markov property} on $\G$ if d-separation implies conditional independence, i.e., if $\cS$ d-separates $X_i, X_j$, then $X_i \ind X_j | \cS$. If on the other hand conditional independence implies d-separation, we say that the distribution $P_{X_1,\dots,X_n}$ is \emph{faithful} to $\G$. A causal model is said to be \emph{causally sufficient} if there are no hidden confounders are present. 
\\ \\
\noindent Recall also that the \emph{moralized graph} of a DAG is the undirected graph over the same set of nodes in which any node is connected to its DAG-parents, -children and any parent of its children.



\subsection{Proofs}

We will start with the proof of Lemma \ref{thm.principle1}. Corollary \ref{thm.ident1} then follows directly from Lemma \ref{thm.principle1}.

\begin{proof}[Proof of Lemma \ref{thm.principle1}]
We prove the theorem by contradiction. Suppose that there is a subset $\cS \subset \Y$ and nodes $Y_k, Y_{\ell} \in \Y$ such that
\begin{align*}
    Y_k \ind Y_{\ell} \ | \ \cS \qquad \text{and} \qquad Y_k \centernot \ind Y_{\ell} \ | \ \cS, \X.
\end{align*}
By  faithfulness every path between $Y_k$ and $Y_{\ell}$ is blocked by $\cS$ and by the Markov property, there must be a path between $Y$ and $\Tilde{Y}$ that is unblocked when conditioning on $\X$ and that therefore has to pass through $\X$. Any such path has to be of the form 
\begin{align*}
    Y - \dots Y' \leftarrow X' - \dots - X'' \rightarrow Y'' - \dots - \Tilde{Y}, 
\end{align*}
as there are only causal arrows from $\X$ to $\Y$ by (A2). Since $X', X''$ cannot be colliders, this path has to be blocked by $\X$, which is a contradiction.
\end{proof}

\noindent Next, we turn to the graphical characterization of Condition (C1) of Lemma \ref{thm.ident1}.


\begin{lemma}[Graphical characterization of (C1)] \label{lem.graphicalcharacterization}
Assume the assumptions of Model 1 to be satisfied. Then (C1) holds iff there is a subset $\cS \subset \X$, scalar variables $X_i,X_j \in \X$ and a path $p$ between $X_i$ and $X_j$ such that
\begin{itemize}
\item[(i)] $\cS$ d-separates $X_i,X_j$ and
\item[(ii)] $p$ passes through $\Y$,
\item[(iii)] both neighbors of any node $Y \in \Y$ on $p$ lie in $\X$, i.e. around $\Y$-variables, $p$ is of the form $X_k \rightarrow Y \leftarrow X_l$,
\item[(iv)] any subpath of $p$ contained in $\X$ is unblocked by $\cS$.
\end{itemize}
\end{lemma}


\begin{figure*}[!ht]
    \centering
\begin{subfigure}[t]{0.2\textwidth}
        \centering
    \includegraphics[scale=0.25]{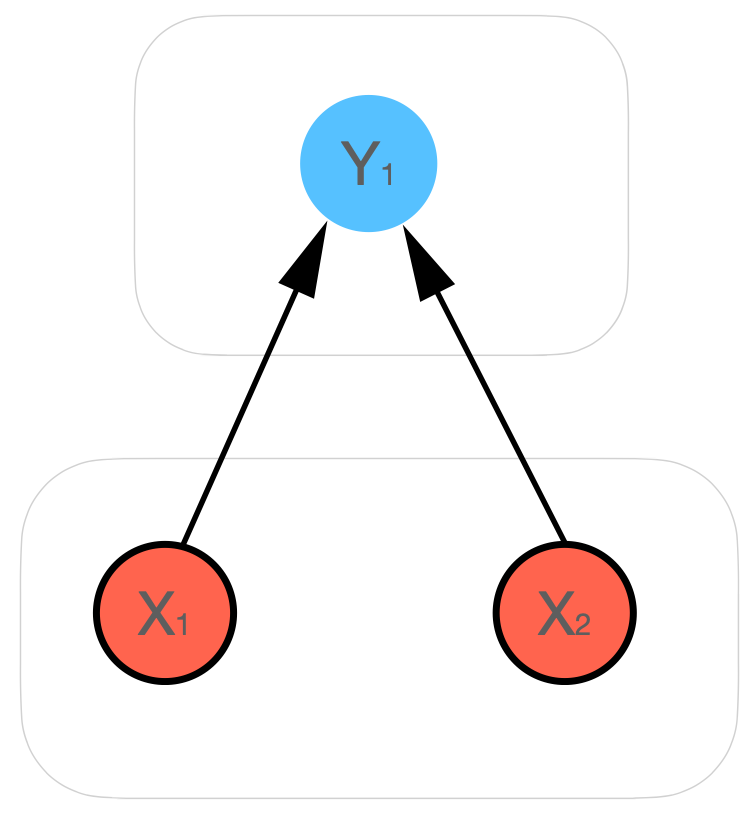}
    \caption{(C1) satisfied}
    \label{fig:lemma3_1a}  
\end{subfigure}%
~
\begin{subfigure}[t]{0.3\textwidth}
        \centering
    \includegraphics[scale=0.25]{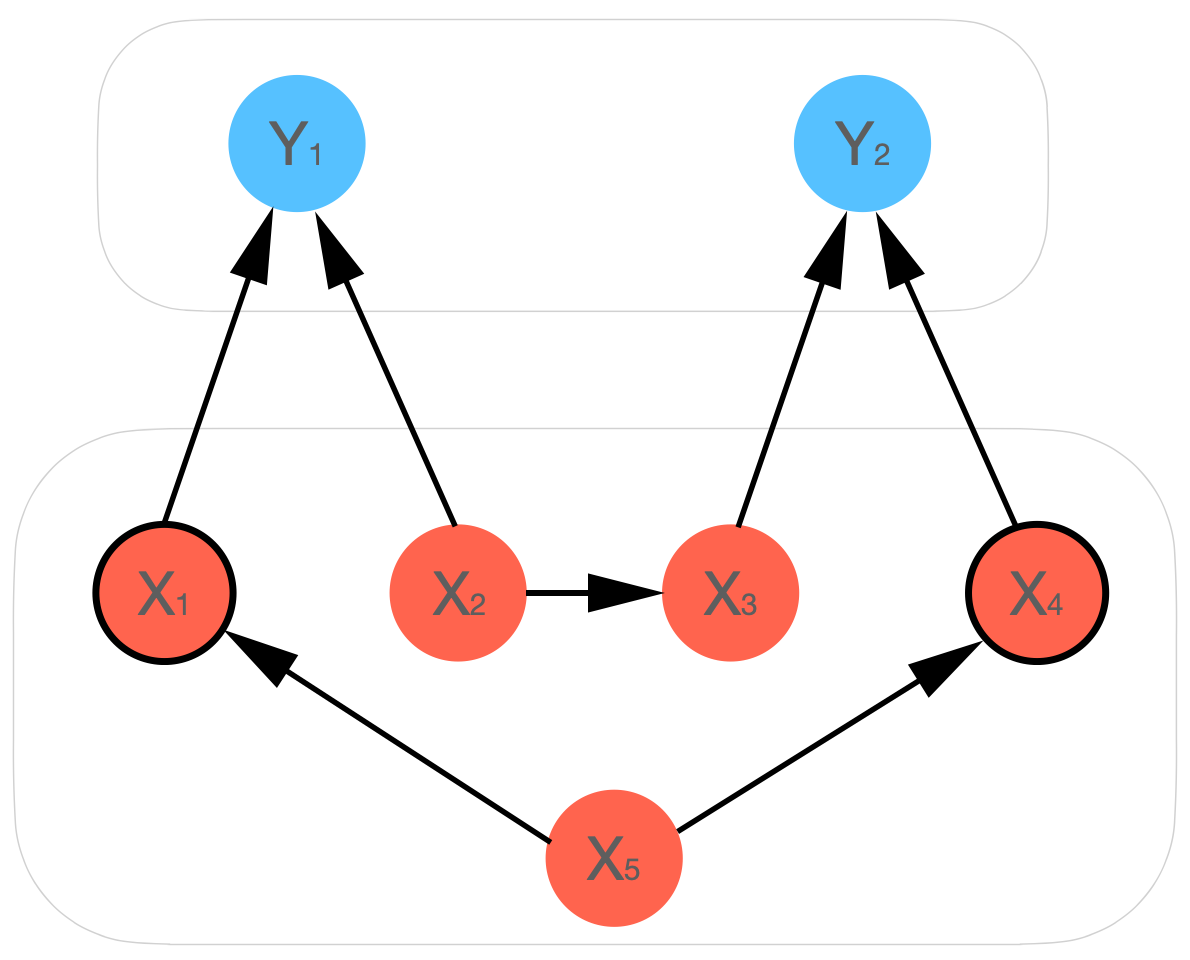}
    \caption{(C1) satisfied}
    \label{fig:lemma3_1b}
\end{subfigure}%
~
\begin{subfigure}[t]{0.2\textwidth}
        \centering
    \includegraphics[scale=0.25]{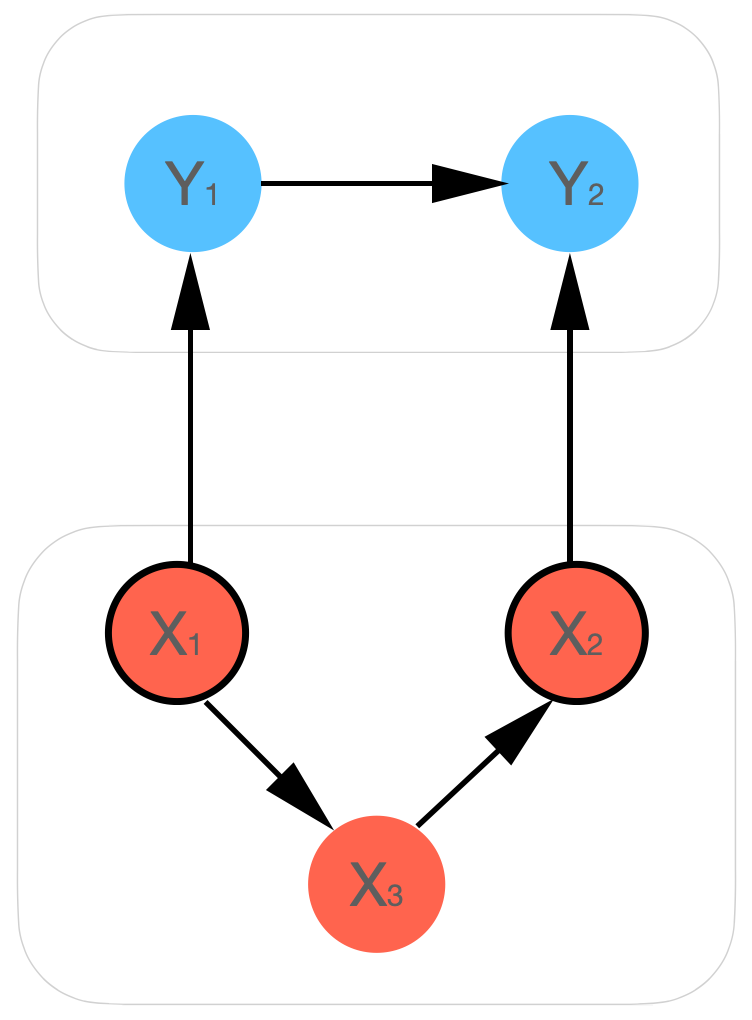}
    \caption{(C1) not satisfied }
    \label{fig:lemma3_1c}
\end{subfigure}%
~
\begin{subfigure}[t]{0.2\textwidth}
        \centering
    \includegraphics[scale=0.25]{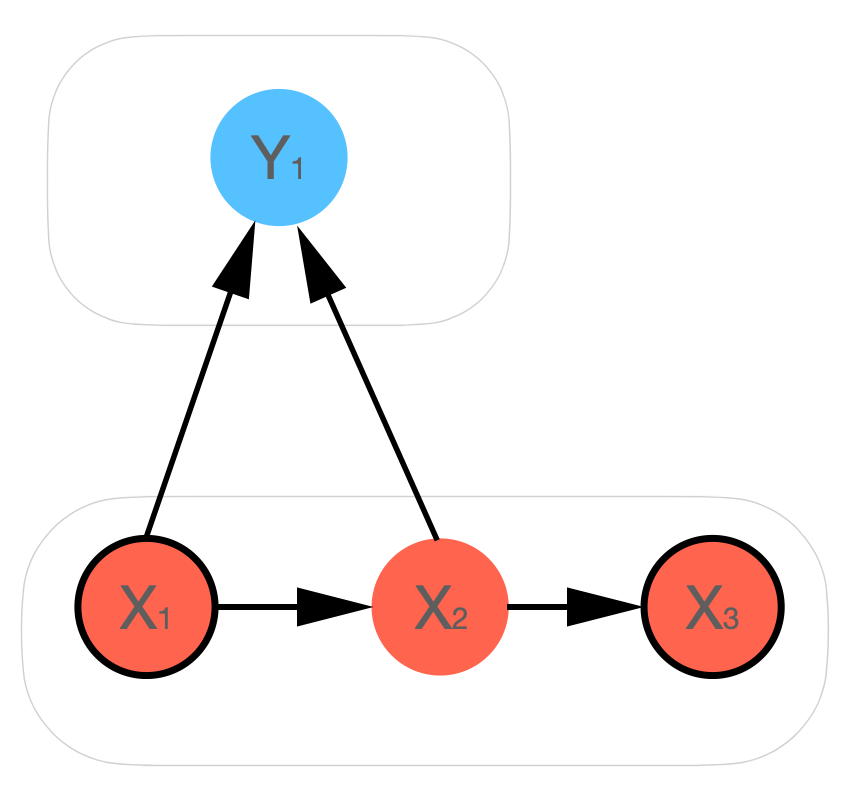}
    \caption{(C1) not satisfied}
    \label{fig:lemma3_1d}
\end{subfigure}
\caption{Examples of the unidirectional causal vector model $\X \to \Y$ for vector-valued variables $\X$ and $\Y$ for illustrating Lemma \ref{lem.graphicalcharacterization}. Cases (a), (b) satisfy (C1) for $X_i$ and $X_j$ denoted by nodes with a bold border. In particular, in (a) $X_1 \ind X_2 | \cS $, where $\cS = \emptyset$ and $X_1 \centernot \ind X_2 | \cS, \Y $ and, in (b) $X_1 \ind X_4 | \cS $, where $\cS = X_5$ and $X_1 \centernot \ind X_4 | \cS, \Y $. Cases (c), (d) do not satisfy (C1) for $X_i$ and $X_j$ denoted by nodes with a bold border and hence Corollary \ref{thm.ident1} does not hold. In particular, in (c) $X_1 \ind X_2 | X_3$ but $X_1 \ind X_2| X_3,\Y$ since point (iii) of Lemma \ref{lem.graphicalcharacterization} doesn't hold. Similarly in (d) point (iv) of Lemma \ref{lem.graphicalcharacterization} doesn't hold and $X_1 \ind X_3| X_2,\Y$.} \label{fig.lemma3_graph}
\end{figure*}

In Figure \ref{fig.lemma3_graph}, the conditions of Lemma \ref{lem.graphicalcharacterization} are depicted with the help of illustrative examples.
\begin{proof}[Proof of Lemma \ref{lem.graphicalcharacterization}]
Suppose that there is a subset $\cS$ of $\X$ that d-separates $X_i,X_j$ and a path $p$ as in the lemma. By the Markov property, we have $X_i \ind X_j | \cS$. Moreover, any $Y \in \Y$ that lies on $p$ is a collider so that the path is unblocked by the set of all such colliders and hence by $\Y$. Faithfulness thus implies that $X_i \centernot \ind X_j | \cS, \Y$. 
\\
Conversely if (C1) holds, by Faithfulness there must be $\cS\subset \X$ that d-separates $X_i,X_j$ (proving (i)) and a path $q$ between $X_i,X_j$ that is unblocked by $\cS,\Y$ and that hence must pass through $\Y$ (proving (ii)). Since $q$ is unblocked, any subpath within $\X$ must be unblocked by $\cS$ which proves (iv). If there was $Y$ on $q$ that had a neighbor $\tilde{Y} \in \Y$ on $q$, we claim that $q$ could not be unblocked by $\Y$. Indeed, since $Y$ and $\tilde{Y}$ cannot both be colliders, conditioning on $\Y$ would block $q$. Thus, (iii) holds as well.
\end{proof}

\noindent We now prove Lemma \ref{thm.principle2} from which Corollary \ref{thm.ident2} and thus Theorem \ref{thm.main} follows immediately.

\begin{proof}[Proof of Lemma \ref{thm.principle2}]
Because of (A2) any path between $X_i, X_j \in \X$ that passes through $\Y$ must contain a collider in $\Y$ and must therefore be blocked. If there existed $\cS \subset \X$ as in the theorem, then by the Markov property there must be a path between $X_i$ and $X_j$ that is unblocked by $\cS$. Hence it cannot pass through $\Y$. Therefore this path is still unblocked by $\cS \cup \Y$. By Faithfulness, it follows that $X_i \centernot \ind X_j | \cS, \Y$. This contradicts our assumption that $X_i \ind X_j | \cS, \Y$.
\end{proof}

\noindent Graphically, condition (C2) can be characterized as follows.

\begin{lemma}[Graphical characterization of (C2)] \label{lem.graphicalcharacterization2}
Assume the assumptions of Model 1 to be satisfied. Then (C2) holds iff there is a subset $\cS \subset \Y$, scalar variables $Y_k,Y_l \in \Y$ and a path $p$ between $Y_k$ and $Y_l$ such that
\begin{itemize}
    \item[(i)] $\cS$ d-separates $Y_k,Y_l$ in the restriction of the DAG $\G$ to $\Y$, i.e. in the subgraph that contains only nodes of  $\Y$ and edges between them.
    \item[(ii)] $p$ is unblocked by $\cS$ in $\G$ and passes through $\X$.
\end{itemize}
\end{lemma}

\begin{proof}[Proof of Lemma \ref{lem.graphicalcharacterization2}]
Suppose that (C2) holds. Since $Y_k \ind Y_l | \cS, \X$, by Faithfulness $\cS$ and $\X$ together d-separate $Y_k, Y_l$ in $\G$ and hence $\cS$ must d-separate them in the subgraph over $\Y$. Hence (i) holds. On the other hand, since $Y_k \centernot\ind Y_l | \cS$, by the Markov property $\cS$ does not d-separate the two nodes in the full graph $\G$. Hence there must be a path passing through $\X$ that is not blocked by $\cS$ so that (ii) holds. 

\noindent Conversely, assume that a set $\cS$ and a path $p$ as in the lemma exist. Note that any path between $Y_k, Y_l$ that passes through $\X$ must always be blocked by $\X$ because of the unidirectionality assumption (A2) which implies that the path cannot contain cross-regional colliders of the form $Y_s \rightarrow X_i \leftarrow Y_t$. Since by (i), $\cS$ also blocks any open path internal to $\Y$, by the Markov property, we have $Y_k \ind Y_l | \cS, \X$. Finally by (ii) and Faithfulness, we must have $Y_k \centernot\ind Y_l | \cS$ so that (C2) holds. 
\end{proof}

\begin{figure*}[!ht] 
    \centering
\begin{subfigure}[t]{0.25\textwidth}
        \centering
    \includegraphics[scale=0.25]{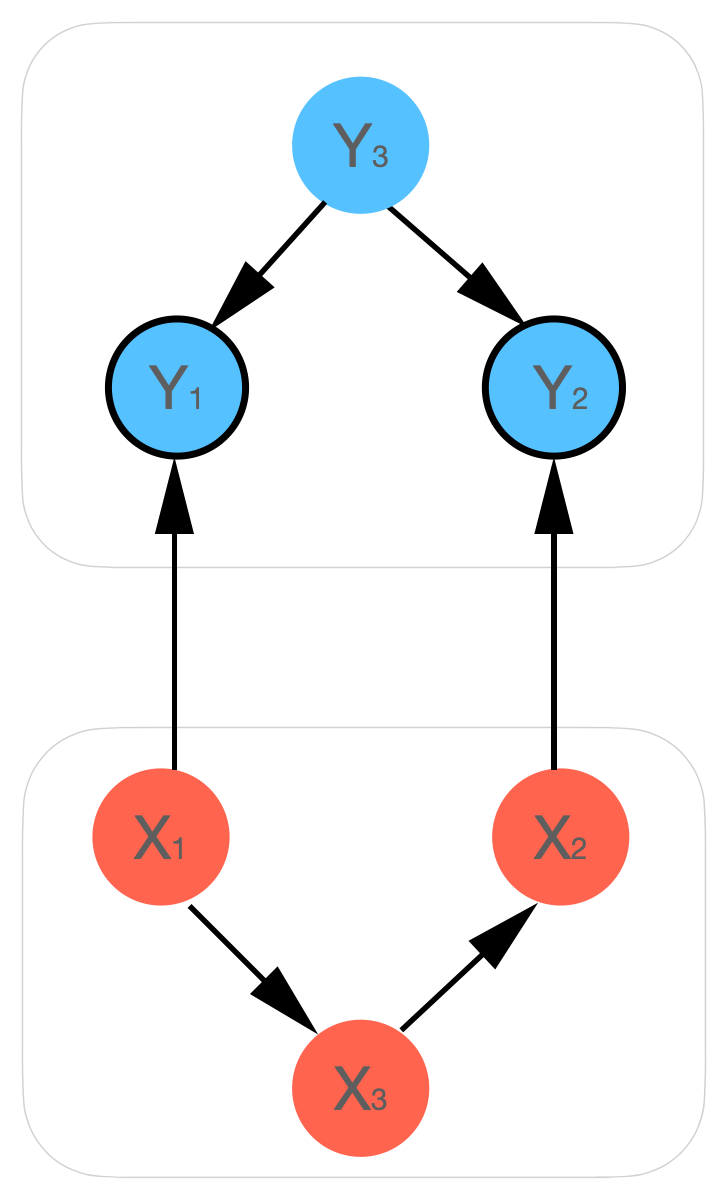}
    \caption{(C2) satisfied}
    \label{fig:lemma3_2a}
\end{subfigure}%
~
\begin{subfigure}[t]{0.25\textwidth}
        \centering
    \includegraphics[scale=0.25]{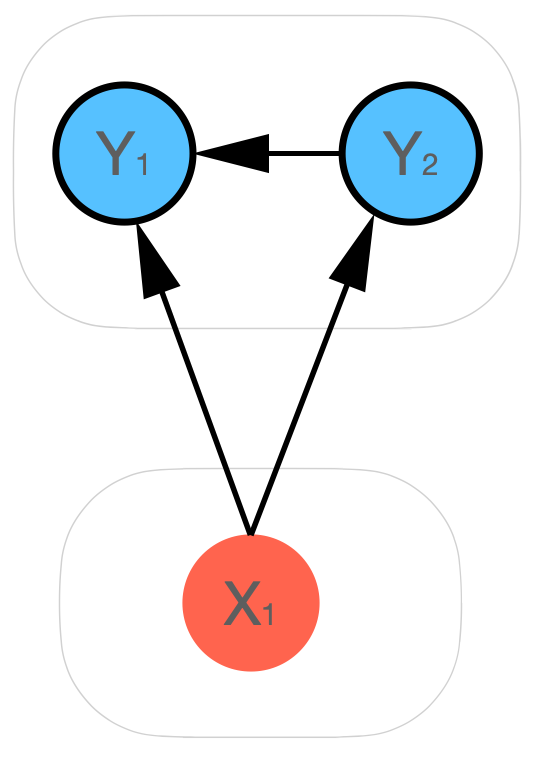}
    \caption{(C2) not satisfied}
    \label{fig:lemma3_2b}
\end{subfigure}%
~
\begin{subfigure}[t]{0.25\textwidth}
        \centering
    \includegraphics[scale=0.25]{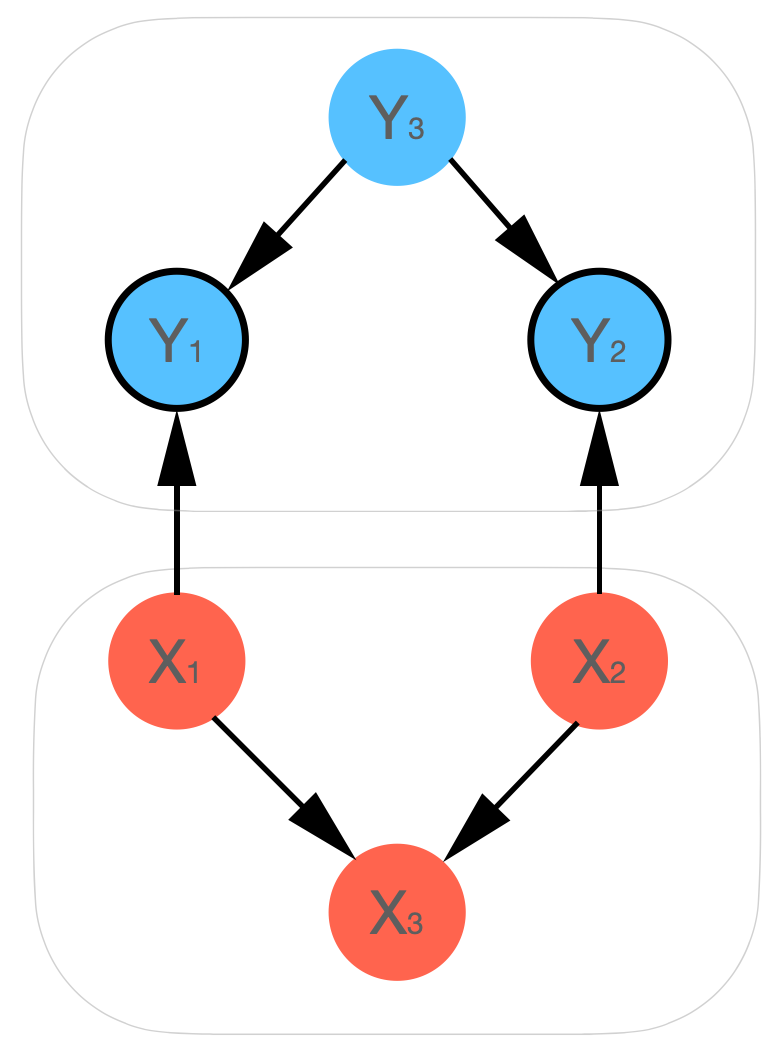}
    \caption{(C2) not satisfied}
    \label{fig:lemma3_2c}
\end{subfigure}
\caption{Examples of the unidirectional causal vector model $\X \to \Y$ for vector-valued variables $\X$ and $\Y$ for illustrating Lemma \ref{lem.graphicalcharacterization2}. In (a), (C2) is satisfied for $Y_k$ and $Y_l$ with a bold border, whereas in (b) and (c) it is not satisfied. In particular, in (a) $Y_1 \centernot \ind Y_2 | Y_3$ and $Y_1 \ind Y_2 | Y_3, \X$. In (b), $Y_1 \centernot \ind Y_2|\cS$ where $\cS = \emptyset$ but $Y_1 \centernot \ind Y_2|\cS,\X$ because point (i) of Lemma \ref{lem.graphicalcharacterization2} is not satisfied. In (c) Point (ii) of Lemma \ref{lem.graphicalcharacterization2} is not satisfied because the path $Y_1 \leftarrow X_1 \to X_3 \leftarrow X_2 \to Y_2$ is not unblocked by $\cS = Y_3$ in $\G$. } \label{fig.lemma4_graph}
\end{figure*}

In Figure \ref{fig.lemma4_graph}, the conditions of Lemma \ref{lem.graphicalcharacterization2} are depicted with the help of illustrative examples.

\begin{proof}[Proof of Theorem \ref{thm.edgedensity1}]
Since by Lemmas \ref{thm.principle1} and \ref{thm.principle2}, principles (P1) and (P2) hold, it follows that $d(\X|\Y) \geq 0$ and $d(\Y|\X) \leq 0$ and thus $d(\X|\Y) - d(\Y|\X) \geq 0$. Condition (C1) then implies that the first inequality becomes a strict inequality $d(\X|\Y) > 0$ and similarly Condition (C2) implies that the second inequality becomes  strict, i.e. $d(\Y|\X) < 0$. Therefore, if one of these conditions holds, then 
\[d(\X|\Y) - d(\Y|\X) > 0. \]
Thus if the causal direction is unknown, it can be inferred from the sign of  $d(\X|\Y) - d(\Y|\X)$.
\end{proof}

\subsection{Comparison with Vanilla-PC algorithm} \label{subsec.compPC}
In order to determine the causal direction of Model 1 (with linear interactions) using an ad-hoc adaptation of the PC algorithm (Vanilla PC), we treat each node in the vector-valued variables $\X$ and $\Y$ as an independent node. We run the PC algorithm on the set of $n+m$ nodes, where $n$ and $m$ are the group sizes of $\X$ and $\Y$, respectively. Recall that the maximum number of edges from $\X$ to $\Y \ \mathrm{edgeMax}$ is $n \cdot m$. On the resulting CPDAG of the of the true graph $\G$, we compute $\mathrm{edge_{X \to Y}}$, i.e. the number of directed edges from nodes in $\X$ to nodes in $\Y$. Similarly we compute $\mathrm{edge_{Y \to X}}$. We then normalise these quantities w.r.t.~$\mathrm{edgeMax}$ to compute,
\begin{equation*}
\mathrm{edgeDens_{X \to Y}}  = \frac{\mathrm{edge_{X \to Y}}}{\mathrm{edgeMax}} \ , \ \mathrm{edgeDens_{Y \to X}}  = \frac{\mathrm{edge_{Y \to X}}}{\mathrm{edgeMax}} \ .
\end{equation*}
Given the sensitivity parameter $\alpha$ and 
\begin{align*}
\mathrm{edgeDiff} = \mathrm{edgeDens_{X \to Y}} - \mathrm{edgeDens_{Y \to X}},    
\end{align*} 
we make a decision according to the following rule:
\begin{align*}
&\text{if } \ |\mathrm{edgeDiff}| \leq \alpha,  &\text{ then `direction indeterminate'}, \\
&\text{if } \  \mathrm{edgeDiff} > \alpha, &\text{ then } `\X \text{ causes } \Y\text{'}, \\
&\text{if } \  \mathrm{edgeDiff} < -\alpha, &\text{ then } `\Y \text{ causes } \X \text{'}. 
\end{align*}
Comparison plots between our methods and Vanilla PC can be found in Figures \ref{fig.groupsize30}, \ref{fig.different_group_sizes_low_group_density}, \ref{fig.different_group_sizes_low_interaction_density} and \ref{fig.different_sample_sizes}.


    \label{fig:raw_pc_comp}


\subsection{Real Data Example Details}

For testing our algorithms on the climate science example considered in the main text, we also plotted a histogram of $\mathrm{Crit} = \widehat{d(\X|\Y)} - \widehat{d(\Y|\X)}$ values, see \eqref{eq.d_1}, \eqref{eq.d_2} for 2G-VecCI.PC and \eqref{eq.d_hat_1}, \eqref{eq.d_hat_2} for 2G-VecCI.Full. This helps visualise the spread of the fraction of correct and wrong inferences across different choices of coarse-graining of the data, see Figure \ref{fig.hist_realdata} (for further details see the \texttt{real\_data\_AAAI.py} file in the Supplement). 


\begin{figure*}[!ht]
\begin{center}
\scalebox{0.9}{
\begin{tabular}{|c|c|c|c|}
\hline
    \textbf{Method} & \textbf{Assumptions} & \textbf{Strengths} & \textbf{Weaknesses} \\
    \hline
    \makecell{Trace Method \\ \citep{Janzing09}.} & \makecell{Linearity; \\ Additive noise. \\ technical assumptions.} & \makecell{Very Fast; \\ Strong empirical performance \\ on linear data \\ for large groups.} & \makecell{ No theoretical guarantees \\ beyond the noiseless case; \\ Not applicable to \\ nonlinear interactions.} \\
    \hline 
    \makecell{LiNGaM on group means \\ \citep{ShimizuLinGaM}.} & \makecell{Linearity; \\ Additive noise; \\ Non-Gaussian noise. } & \makecell{Fast and simple; \\ Reliable for small groups \\ if assumptions are met.\\ } & \makecell{Averaging drives noise \\ close to Gaussian \\ when groups are large. \\ Vulnerable to \\ opposing effects.}  \\
    \hline 
    Vanilla-PC & \makecell{Markov property and \\ Faithfulness \\ on micro-DAG; \\ Assumptions on CI \\ tests.} & \makecell{Non-parametric; \\ Infinite and finite sample guarantees \\
    \citep{KaBue07}; \\ Good performance \\ for small groups.} & \makecell{Slow for \\ dense large groups; \\ Impaired performance \\ for large groups and \\ dense interactions.} \\
    \hline
    2G-VecCI.PC & \makecell{Markov property and \\ Faithfulness \\ on micro-DAG \\ for soundness. \\ (C1), (C2) \\ for completeness. \\ Assumptions \\ on CI tests.}& \makecell{Non-parametric; \\ Infinite sample guarantees; \\ Good performance \\ for large groups \\ in many regimes; \\ Less CI tests than Vanilla-PC \\ in worst case scenarios.} &  \makecell{Slow for\\ large groups; \\ Impaired performance \\ when groups are very small \\ or densely connected.}\\
    \hline
    2G-VecCI.Full & \makecell{Semi-causal or \\ causal model; \\ Principles (P1) and (P2); \\ Assumptions \\ on CI tests.} &  \makecell{Non-parametric; Faster than \\ PC-based methods; \\ Bounded number \\ of CI-tests; \\ Strong empirical \\ performance. }& \makecell{No theoretical guarantees; \\ Impaired performance \\ when groups are very small.} \\
    \hline
\end{tabular}
}
\caption{Existing methods to identify the causal relationship between two groups of variables with unidirectional interactions.} \label{table.comparison}
\end{center}
\end{figure*}

\begin{figure*}[!ht] 
    \centering
    \begin{subfigure}[t]{0.22\textwidth}
        \centering
        \includegraphics[scale=0.3]{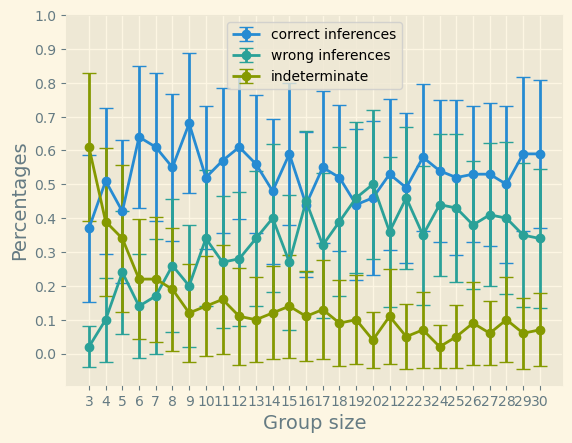}
        \caption{2G-VecCI.PC}
    \end{subfigure}%
    ~ 
    \begin{subfigure}[t]{0.22\textwidth}
        \centering
\includegraphics[scale=0.3]{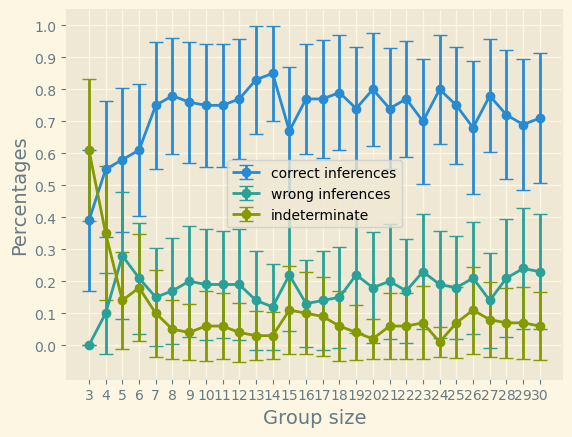}
        \caption{2G-VecCI.Full}
    \end{subfigure}
~
    \begin{subfigure}[t]{0.22\textwidth}
        \centering
        \includegraphics[scale=0.3]{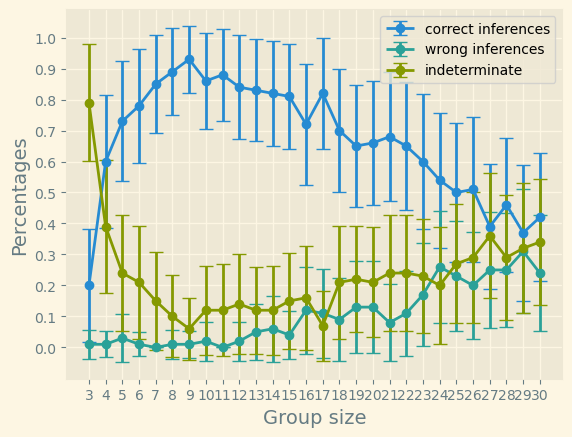}
        \caption{Vanilla-PC}
    \end{subfigure}
    \caption{Performance of 2G-VecCI.PC, 2G-VecCI.Full and Vanilla-PC for groups of different sizes and increased densities ($30 \%$). The density of the interaction matrix $A$ is lowered to $30\%$. Non-zero entries of $A$ are drawn uniformly randomly from $[-0.7,0.7]$ and 100 random models are run per parameter combination with 100 samples each. We set the sensitivity parameter of both 2G-VecCI.PC and Vanilla PC to $10^{-4}$. Vanilla PC deals well with sparse interaction matrices and outperforms our methods when groups are small.}  \label{fig.different_group_sizes_low_interaction_density}
\end{figure*}

\begin{figure*}[!ht] 
    \centering
    \begin{subfigure}[t]{0.22\textwidth}
        \centering
        \includegraphics[scale=0.3]{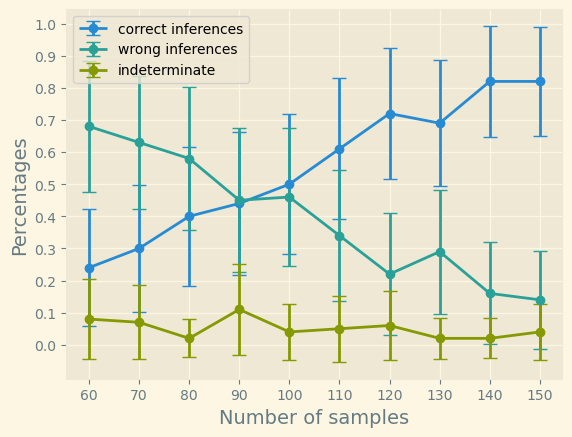}
        \caption{2G-VecCI.PC}
    \end{subfigure}%
    ~ 
    \begin{subfigure}[t]{0.22\textwidth}
        \centering
        \includegraphics[scale=0.3]{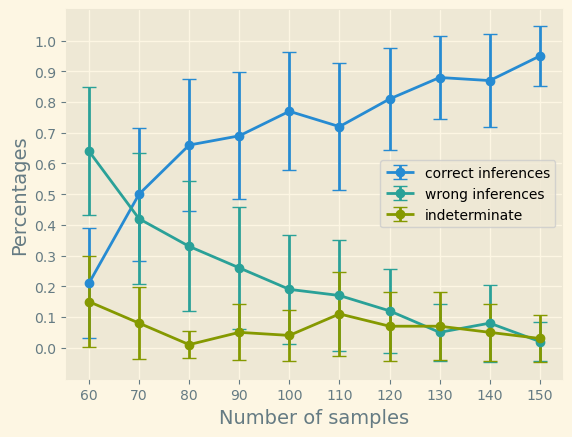}
        \caption{2G-VecCI.Full}
    \end{subfigure}
~
    \begin{subfigure}[t]{0.22\textwidth}
        \centering
        \includegraphics[scale=0.3]{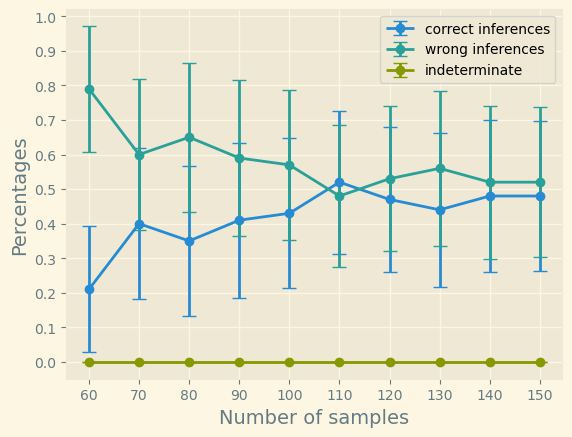}
        \caption{Vanilla-PC}
    \end{subfigure}
    \caption{Performance of 2G-VecCI.PC, 2G-VecCI.Full and Vanilla-PC for groups of size 30 at different sample sizes. Internal group densities and the density of the interaction matrix $A$ are set to $30\%$. Non-zero entries of $A$ are drawn uniformly randomly from $[-0.7,0.7]$ and 100 random models are run per parameter combination. We set the sensitivity parameter of both 2G-VecCI.PC and Vanilla PC to $10^{-5}$ to ensure that the lack of performance is not due to an overly conservative choice of this parameter.} \label{fig.different_sample_sizes}
\end{figure*}

\begin{figure*}[!ht] 
    \centering
    \begin{subfigure}[t]{0.22\textwidth}
        \centering
        \includegraphics[scale=0.3]{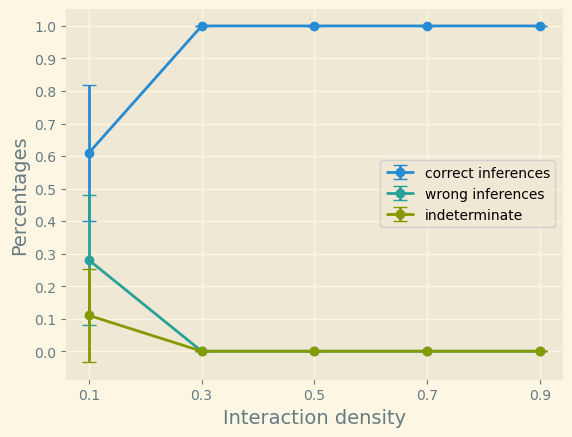}
        \caption{2G-VecCI.Full}
    \end{subfigure}
    ~ 
    \begin{subfigure}[t]{0.22\textwidth}
        \centering
        \includegraphics[scale=0.3]{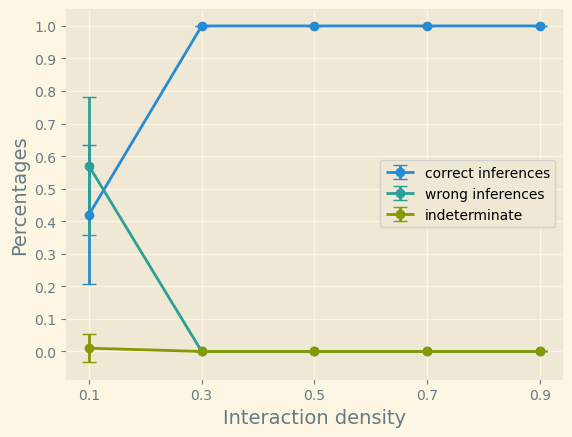}
        \caption{Trace Method}
    \end{subfigure}
    \caption{Performance of 2G-VecCI.Full and the trace methods for groups of size 100 with 500 available samples and linear interactions. Performance is shown along increasing density of the interaction matrix $A$. Non-zero entries of $A$ are drawn uniformly randomly from $[-0.7,0.7]$, groups are assumed dense ($\approx 30$ connections for every node) and 100 random models are run per choice of parameters. Both methods recover the correct causal direction well.}
    \label{fig.groupsize100}
\end{figure*}
\begin{figure*}[!ht] 
    \centering
    \begin{subfigure}[t]{0.22\textwidth}
        \centering
    \includegraphics[scale=0.3]{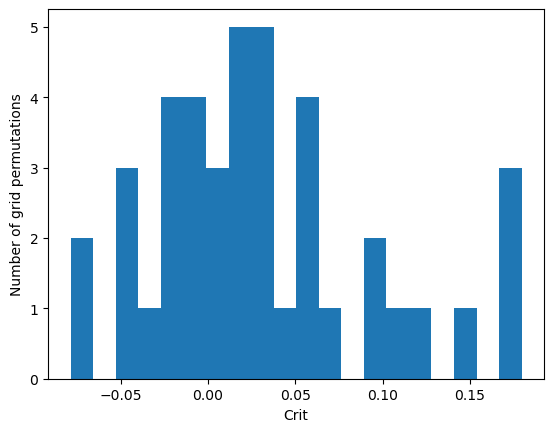}
    \caption{2G-VecCI.Full}
    \end{subfigure}%
    ~
       \begin{subfigure}[t]{0.22\textwidth}
        \centering
    \includegraphics[scale=0.3]{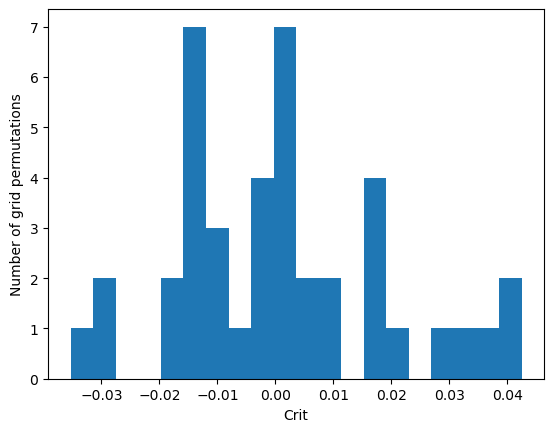}
    \caption{2G-VecCI.PC}
    \end{subfigure}
    \caption{Histogram of $\mathrm{Crit}$ values for different choices of coarse grainings of Nino and BCT surface temperatures. In (a) there is a trend favouring $\mathrm{Crit}>0$. In particular, out of the 41 different coarse grainings considered, 26 have $\mathrm{Crit}>0$, 14 have $\mathrm{Crit}<0$, 1 has $\mathrm{Crit}=0$ and positive values tend to be higher. In (b), however, there is no such clear trend and, in particular, out of the 41 different coarse grainings considered, 18 have $\mathrm{Crit}>0$, 20 have $\mathrm{Crit}<0$ and 3 have $\mathrm{Crit}=0$.}
    \label{fig.hist_realdata}
\end{figure*}

\end{document}